\documentclass[a4paper,UKenglish,cleveref, autoref, thm-restate, nolineno]{socg-lipics-v2021}

\pdfoutput=1 
\hideLIPIcs  

\usepackage[ruled,vlined,linesnumbered]{algorithm2e}

\title{Asymptotically Robust Learning-Augmented Algorithms for Preemptive FIFO Buffer Management}

\author{Wen-Han Hsieh}{Department of Computer Science and Information Engineering, National Cheng Kung University,
Tainan, Taiwan}{p76131165@gs.ncku.edu.tw}{https://orcid.org/0009-0006-3103-714X}{} 

\author{Ya-Chun Liang}{Department of Computer Science, National Tsing Hua University, Hsinchu, Taiwan}{ycliang@cs.nthu.edu.tw}{https://orcid.org/0000-0002-7359-4335}{}

\authorrunning{Wen-Han Hsieh and Ya-Chun Liang} 

\Copyright{Wen-Han Hsieh and Ya-Chun Liang} 



\EventEditors{John Q. Open and Joan R. Access}
\EventNoEds{2}
\EventLongTitle{42nd Conference on Very Important Topics (CVIT 2016)}
\EventShortTitle{CVIT 2016}
\EventAcronym{CVIT}
\EventYear{2016}
\EventDate{December 24--27, 2016}
\EventLocation{Little Whinging, United Kingdom}
\EventLogo{}
\SeriesVolume{42}
\ArticleNo{23}

\newcommand{\ALG}{\mathrm{ALG}}
\newcommand{\OPT}{\mathrm{OPT}}
\newcommand{\PG}{\mathrm{PG}}

\begin{document}

\maketitle

\begin{abstract}
We present a learning-augmented online algorithm for the preemptive FIFO buffer management problem,
where packets arrive online to a finite-capacity buffer, must be transmitted in FIFO order, and the algorithm may preemptively discard buffered packets to accommodate future arrivals. Our algorithm
simultaneously achieves
$1$-consistency, $\eta$-smoothness, and asymptotic $\sqrt{3}$-robustness, where $\eta$ denotes the prediction error. 
Specifically, it attains an optimal competitive ratio of $1$ under perfect predictions, degrades smoothly as the prediction error increases, and maintains an asymptotic competitive ratio of $\sqrt{3}$ under arbitrarily inaccurate predictions, matching the best-known worst-case guarantee for the classical online problem~\cite{englert2009lower}.

A key technical contribution of our work is the introduction of an \emph{output-based prediction error metric}. Because capacity constraints dictate that only a strictly bounded subset of arriving packets is ultimately transmitted, our metric assesses prediction quality over the resulting optimal schedules rather than the raw input sequences, avoiding artificial error penalties. 
To guarantee robustness, our algorithm dynamically monitors predictions and executes a \emph{buffer-clearing strategy} upon transitioning to a worst-case fallback mechanism. We prove that the competitive loss incurred by this clearing operation is bounded by an additive capacity constant that vanishes asymptotically. Finally, we show that our algorithm provides a generalized framework for learning-augmented buffer management: substituting the fallback module with any $\beta$-competitive online algorithm immediately yields asymptotic $\beta$-robustness.
\end{abstract}

\section{Introduction} 
\label{sec1:introduction}
Efficient network traffic management is essential for high-performance applications, such as video conferencing and cloud computing, which require high data throughput, low latency, and 
strict
reliability. In such 
architectures,
Quality of Service (QoS) switches play a central role by offering differentiated services 
across
various traffic categories. 
System
performance relies on two fundamental mechanisms: \emph{buffer management} and \emph{packet scheduling}. The former operates under spatial constraints:
packets arrive with associated values, and the system must make instantaneous admission
and eviction
decisions under a fixed buffer capacity to maximize the total transmitted value. In contrast, packet scheduling focuses on temporal constraints, 
typically 
modeled as scheduling on an unbounded buffer subject to 
release times and deadlines.

We study the \emph{preemptive FIFO buffer management problem}. 
While
numerous 
variants of packet scheduling and buffer management
have been extensively studied,
the competitive ratio for this 
specific
problem has remained at $\sqrt{3}$ since the 
seminal
work of Englert and Westermann~\cite{englert2009lower}, 
leaving a persistent gap with the known lower bound of $1.419$.
Furthermore, 
the design of learning-augmented algorithms for this setting remains entirely unexplored.
Our objective is to bridge this gap by integrating machine-learned predictions
into the decision-making process.
We aim to achieve optimal performance under 
perfectly
accurate predictions while 
strictly
maintaining the best-known $\sqrt{3}$-competitiveness as a worst-case guarantee, regardless of prediction quality.


\subsection{Our Results} 
\label{sec1-1:our results}
\begin{itemize}
    \item We design a learning-augmented online algorithm (Algorithm~\ref{alg:main algorithm}) that achieves an optimal competitive ratio of 1 under perfect predictions (Theorem~\ref{thm:main consistency}) and degrades smoothly as the prediction error $\eta$ increases (Theorem~\ref{thm:main smoothness}). Even when predictions are arbitrarily inaccurate, the algorithm maintains an asymptotic competitive ratio of $\sqrt{3}$ (Theorem~\ref{thm:main robustness}).

    \item We 
    introduce
    an 
    \emph{output-based prediction error}
    (Definition~\ref{def:prediction error}), defined as the total value of the symmetric difference between the 
    offline
    optimal schedules under the true and predicted inputs. This metric 
    fundamentally departs from traditional input-based error measures. It
    is motivated by the observation that
    in admission control, untransmitted packets do not impact the final objective; thus,
    a well-defined error metric must 
    penalize only those 
    discrepancies that 
    explicitly alter the transmitted schedule.

    \item Our algorithm incorporates a 
    dynamic
    protective fallback mechanism that switches to 
    a robust baseline
    upon detecting an unreliable prediction. Since the buffer state at the time of switching was 
    constructed under a distrusted prediction,
    inheriting it 
    directly could compromise future performance guarantees.
    To resolve this, we propose
    a 
    \emph{buffer-clearing strategy}
    upon 
    transition.
    We
    prove that the resulting 
    competitive
    loss is bounded by a 
    strict
    constant 
    (Lemma~\ref{lem:OPT comparison}) that vanishes asymptotically as the total transmitted value grows (Theorem~\ref{thm:main robustness}).

    \item We 
    demonstrate
    that our algorithm 
    provides a generalized
    framework for learning-augmented buffer management (Corollary~\ref{cor:generalized framework}). Any $\beta$-competitive algorithm can 
    serve as the fallback mechanism, immediately
    yielding an asymptotic $\beta$-robustness guarantee.
\end{itemize}

\subsection{Related Work} 
\label{sec1-2:related work}
\subparagraph{The FIFO Buffer Management Problem.}
The problem was initially introduced by Aiello et al.~\cite{aiello2005competitive} and Kesselman et al.~\cite{kesselman2004buffer}. The former restricted packet values to only two 
possible values,
$1$ and $\alpha$ (for some fixed $\alpha > 1$), while the latter examined the general case.
Both studies 
demonstrated
that the \textsc{Greedy} algorithm achieves a competitive ratio of 2. Mansour et al.~\cite{mansour2004optimal} 
later
introduced a lossy scheduling scheme that 
permitting the preemption of 
previously accepted packets. 
The literature subsequently split into \emph{non-preemptive} and \emph{preemptive} models.

For the 
non-preemptive 
model, Andelman et al.~\cite{andelman2003competitive-preemptive} and Zhu~\cite{zhu2004analysis} proposed the \textsc{Ratio} algorithm,
proving 
an optimal competitive ratio of $1 + \ln\alpha$ for the general case, where $\alpha$ is the ratio of the 
maximum to minimum
packet value.

For the 
preemptive
model, Kesselman et al.~\cite{kesselman2005improved} improved 
the upper bound to $1.983$ via the \textsc{Preemptive Greedy($\beta$)} algorithm. 
They also established 
a
lower bound of $\phi = (1+\sqrt{5})/2 \approx 1.618$ for their algorithm, and strengthened the 
general problem
lower bound from $\sqrt{2}\approx1.414$~\cite{andelman2003competitive-preemptive} to 1.419.
Bansal et al.~\cite{bansal2004further} 
later
improved the competitive ratio to $7/4 = 1.75$ 
using
\textsc{Modified-Preemptive Greedy($\beta$)} algorithm. 
Finally, 
Englert and Westermann~\cite{englert2009lower} provided a 
tightened
analysis,
establishing an upper bound of
$\sqrt{3} \approx 1.732$ for a specific parameter $\beta = 2 + \sqrt{3}$,
alongside 
a lower bound of $1 + 1/{\sqrt{2}} \approx 1.707$. 

\begin{table}[ht] 
\caption{
Summary of previous  
results for the FIFO buffer management problem. 
($-$) indicates 
no explicitly stated bound.
}
\label{table:FIFO buffer management summary}
\centering
    \begin{tabular}{llcc}
        \hline\hline
        \textbf{Model} & \textbf{Algorithm} & 
        \textbf{Upper Bound} & \textbf{Lower Bound}\\
        \hline
        \multirow{2}{*}{Non-preemptive} & 
        \textsc{Greedy}~\cite{aiello2005competitive, kesselman2004buffer} & 
        2 & 1.618 \\
        & \textsc{Ratio}~\cite{andelman2003competitive-preemptive} & 
        $1+\ln\alpha$ & $1+\ln\alpha$\\
        \hline
        \multirow{4}{*}{Preemptive} & 
        \textsc{Greedy}~\cite{aiello2005competitive, kesselman2004buffer} & 
        2 & 1.618 \\
        & \textsc{Preemptive Greedy}~\cite{kesselman2005improved} & 
        1.983 & 1.618\\
        & \textsc{Modified Preemptive Greedy}~\cite{bansal2004further} & 
        1.75 & $-$\\
        & \textsc{Preemptive Greedy}~\cite{englert2009lower} & 
        $\sqrt{3} \approx 1.732$ & $1+1/\sqrt{2} \approx 1.707$\\
        \hline\hline
    \end{tabular}
\end{table}

\subparagraph{Online Algorithms with Predictions.}
Learning-augmented algorithms~\cite{lykouris2018competitive, 
purohit2018improving} leverage machine-learned predictions to 
bypass the pessimistic limits of classical
worst-case analysis. The goal is to design 
frameworks that achieve near-optimal performance when predictions are accurate (consistency) while strictly bounding performance degradation when predictions fail (robustness). 
This paradigm has seen extensive application across domains 
including, 
among others,
Scheduling~\cite{purohit2018improving, lindermayr2022permutation, 
im2023non, liang2023learning, balkanski2025scheduling}, 
Caching~\cite{lykouris2018competitive, jiang2022online, 
bansal2022learning}, and the Traveling Salesman 
Problem~\cite{hu2022online, chawla2023online, gouleakis2023learning, hu2025online}. 

A 
critical
design challenge 
is 
selecting a suitable prediction model and 
an appropriate
error metric. 
Prior models
include permutation predictions for input orderings~\cite{lindermayr2022permutation, 
lassota2023minimalistic}, frequency predictions for arrival 
rates~\cite{im2021online, angelopoulos2023online, berg2024online}, and action predictions for offline optimal 
decisions~\cite{lattanzi2020online, antoniadis2023online}. Most 
relevant to our 
methodology,
Liang et al.~\cite{liang2023learning} proposed the \textsc{LAP} algorithm for 
packet scheduling with deadlines,
quantifying 
error via the $\ell_\infty$-norm of the competitive ratio at each time step.


\section{Preliminaries} 
\label{sec2:preliminaries}
\subsection{Problem Definition} 
\label{sec2-1:problem definition}
The 
preemptive FIFO buffer management problem
relaxes the 
strict admission 
constraints of the traditional FIFO model by allowing 
the eviction of previously accepted packets.
The objective is to maximize the total transmitted value subject to a 
strictly
fixed buffer capacity.
We model time as a sequence of discrete 
steps $i\in \{1, 2, \ldots, T\}$. Let $\sigma = (\sigma_1, \sigma_2, \ldots, \sigma_T)$ 
denote
the complete input sequence, where each $\sigma_i = (p_{i_1},p_{i_2},\ldots,p_{i_k})$
is 
the sequence of packets arriving at time step $i$ for some $k \in \mathbb{Z}_+$. 
Each packet $p \in \sigma_i$ carries a value $v(p)$.
Every
time step consists of two sequential phases: 
\begin{enumerate}
    \item {Arrival Phase:} 
    Upon the arrival of each packet, the algorithm 
    must instantaneously decide
    to  
    accept or drop it. The algorithm 
    may preempt (discard) any resident packet to free capacity.
    Each accepted packet is 
    enqueued
    at the tail of the buffer,
    preserving FIFO semantics.

    \item {Transmission Phase.}
    After processing all 
    arrivals
    in $\sigma_i$, exactly one transmission occurs: the packet at the head of the buffer is transmitted, and the system 
    irrevocably accrues its value.
\end{enumerate}

\subsection{Notation} 
\label{sec2-2:notation}
We 
establish notation applicable 
to any algorithm
$\ALG$,
including the offline optimal schedule ($\OPT$) and the \textsc{Preemptive Greedy} algorithm ($\PG$).
We denote the 
final schedule produced by 
$\ALG$ on input $\sigma$ as $\ALG(\sigma)$, and its total transmitted value as $v(\ALG(\sigma)) = \sum_{p \in \ALG(\sigma)} v(p)$. For any time step $t$, 
$v(\ALG_{<t}(\sigma))$ and $v(\ALG_{\le t}(\sigma))$
denote the total value transmitted strictly before and including 
step $t$,
respectively.
Similarly, 
$v(\ALG_{\ge t}(\sigma))$ and $v(\ALG_{> t}(\sigma))$ 
denote
the value transmitted from
step $t$ onward.
By definition:
\begin{equation}
    v(\ALG(\sigma)) = v(\ALG_{<t}(\sigma)) + v(\ALG_{\ge t}(\sigma)) = v(\ALG_{\le t}(\sigma)) + v(\ALG_{> t}(\sigma)). \notag
\end{equation}

For an algorithm $\ALG$
initialized with 
buffer state $B$, we write $\ALG(B, \sigma)$ to denote its schedule.
$B^{\ALG}_t$
represents
the buffer state at the end of time step 
$t$, holding 
total value $v(B^{\ALG}_t)$.
For a learning-augmented algorithm 
utilizing
prediction $\hat{\sigma}$ 
alongside
the actual input $\sigma$, we denote 
the schedule 
as $\ALG(\sigma, \hat{\sigma})$ and 
the associated prediction error as
$\eta(\sigma, \hat{\sigma})$ 
(formalized
in Section~\ref{sec2-4:prediction error}). Similarly, 
the additive property holds 
for any time step $t$, 
i.e.,
\begin{equation}
    v(\ALG(\sigma, \hat{\sigma})) = v(\ALG_{< t}(\sigma, \hat{\sigma})) + v(\ALG_{\ge t}(\sigma, \hat{\sigma})) = v(\ALG_{\le t}(\sigma, \hat{\sigma})) + v(\ALG_{> t}(\sigma, \hat{\sigma})) \notag.
\end{equation}

\subsection{Performance Metrics} 
\label{sec2-3:performance metrics}
\subsubsection{Online Algorithms}
Performance is evaluated 
via competitive analysis~\cite{sleator1985amortized, karlin1988competitive}.
\begin{definition}[Competitive Ratio] 
\label{def:competitive ratio}
    An online algorithm $\ALG$ for a maximization objective is $c$-competitive
    if there exists a constant $k \ge 0$ such that for all input sequences $\sigma$:
    \begin{equation}
         v(\OPT(\sigma))\le c \cdot v(\ALG(\sigma)) + k, \quad \text{where } c = \sup_{\sigma} \frac{v(\OPT(\sigma))}{v(\ALG(\sigma))}.
         \notag
    \end{equation}
\end{definition}

An algorithm is 
strictly $c$-competitive 
if $k=0$, and weakly $c$-competitive otherwise. 
We 
also define
the asymptotic competitive ratio
to capture limiting behavior.
\begin{definition}[Asymptotic Competitive Ratio] \label{def:asymptotic competitive ratio}
    The asymptotic competitive ratio $AC$ 
    is defined as:
    \begin{equation}
        AC = \liminf\limits_{n\to\infty}\left\{\frac{v(\ALG(\sigma))}{v(\OPT(\sigma))} \;\middle|\; v(\OPT(\sigma)) = n\right\} .\notag
    \end{equation}
\end{definition}
\begin{remark} 
\label{rmk:asymptotic equivalent formulation}
    For a maximization objective, the traditional definition of the asymptotic competitive ratio is mathematically equivalent to the following formulation:
    \begin{equation}
        \liminf_{n\to\infty}\left\{\frac{v(\ALG(\sigma))}{v(\OPT(\sigma))} \;\middle|\; v(\OPT(\sigma))=n\right\} \geq \frac{1}{c}\iff\limsup_{v(\ALG(\sigma))\to\infty} \frac{v(\OPT(\sigma))}{v(\ALG(\sigma))} \leq c \notag
    \end{equation}
    In our subsequent proofs, we adopt the 
    right-hand
    formulation
     (evaluating the limit as $v(\ALG(\sigma)) \to \infty$).
     This
     maintains consistency with
     both the standard structure of competitive analysis and the established conventions in prior work on this problem.
\end{remark}

\subsubsection{Learning-Augmented Online Algorithms}
The performance of 
learning-augmented algorithms is characterized by three criteria:
\begin{definition}[Learning-Augmented Properties] \label{def:learning-augmented properties}
    For a maximization problem, an algorithm $\ALG(\sigma, \hat{\sigma})$ is:
    \begin{itemize}
        \item \textbf{$\alpha$-consistent:}
        if 
        \begin{equation}
            \sup_{\hat{\sigma} = \sigma} \frac{v(\OPT(\sigma))}{v(\ALG(\sigma, \hat{\sigma}))} \le \alpha.\quad\text{(performance under perfect predictions)} \notag
        \end{equation}
        \item \textbf{$\beta$-robust:} 
        if
            \begin{equation}
                \sup_{\hat{\sigma}} \frac{v(\OPT(\sigma))}{v(\ALG(\sigma, \hat{\sigma}))} \le \beta.\quad
                \text{(worst-case guarantee)} \notag
            \end{equation}
        \item \textbf{$\eta$-smooth:}    
        if there exists a non-decreasing 
        function $f(\eta)$ such that
            \begin{align}
                 &\frac{\OPT(\sigma)}{\ALG(\sigma, \hat{\sigma})} \le f(\eta(\sigma, \hat{\sigma})),\text{satisfying $f(0) = \alpha$, and 
                 bounded globally by $\beta$.} \notag\\
                 &\text{(graceful degradation of the competitive ratio as prediction error $\eta$ increases)}\notag
            \end{align}
    \end{itemize}
\end{definition}

\subsection{Output-Based Prediction Error} 
\label{sec2-4:prediction error}
A fundamental challenge in 
learning-augmented mechanism design
is 
formulating an error metric that faithfully captures the relationship between prediction accuracy and algorithmic degradation.
Traditionally, many
online
problems adopt input-based error measures, such as the $\ell_1$-norm 
distance
between the predicted and actual input sequences. While 
mathematically intuitive for unconstrained environments, these metrics structurally overestimate prediction error in admission control paradigms like the preemptive FIFO buffer, where the system processes an entire arrival sequence but fundamentally admits and transmits only a strictly bounded subset.

To illustrate this structural mismatch, 
consider a 
severely capacity-constrained 
scenario where
a massive volume of
packets 
arrives
over time, but the buffer 
size permits the transmission of only a tiny fraction (e.g., admitting $100$ out of $1000$ arriving packets).
If 
a prediction mechanism perfectly anticipates the values and optimal schedules for the transmitted subset but wildly misestimates the values of the discarded majority, an input-based metric would accumulate a massive, yet entirely artificial, error penalty. 
Since neither the offline optimal schedule nor the online algorithm ever transmits these discarded packets, such mispredictions exert zero influence on the final objective value. This reveals a critical flaw: penalizing discrepancies over untransmitted packets conflates irrelevant input noise with actual scheduling degradation.

Motivated by this observation, we argue that a well-defined error metric for buffer management must strictly isolate the discrepancies that explicitly alter the final transmitted sequence. 
Consequently, 
we introduce 
a novel
output-based prediction error.
Instead of measuring the distance between input sequences, this metric is formulated over 
the total value of the symmetric difference between the 
offline
optimal schedules 
generated
under the true and predicted inputs. 
By shifting the evaluation from the input stream to the resulting optimal schedules, this metric precisely captures the true cost of prediction failures. 
The following are the
formal
definitions of symmetric difference and our prediction error
metric.
\begin{definition}[Symmetric Difference] 
\label{def:symmetric difference}
    For sets $A$ and $B$, the symmetric difference is
    \(A \triangle B = (A\setminus B)\cup (B\setminus A)\).
\end{definition}
\begin{definition}[Output-Based Prediction Error] \label{def:prediction error}
    Let $\OPT(\sigma)$ and $\OPT(\hat{\sigma})$ be the offline optimal schedules 
    for
    the true sequence $\sigma$ and prediction $\hat{\sigma}$, respectively. The prediction error $\eta(\sigma,\hat\sigma)$ is:
    \begin{equation}
    \label{eq:prediction error}
        \eta(\sigma, \hat{\sigma}) = \sum_{p \in 
        \OPT(\hat{\sigma}) \triangle \OPT(\sigma)} 
        v(p) = \underbrace{\sum_{q \in \OPT(\hat{\sigma}) 
        \setminus \OPT(\sigma)} v(q)}_{D(\hat{\sigma}, 
        \sigma)} + \underbrace{\sum_{q' \in \OPT(\sigma) 
        \setminus \OPT(\hat{\sigma})} v(q')}_{D(\sigma, 
        \hat{\sigma})}. \notag
    \end{equation}
\end{definition}
\begin{remark} \label{remark:prediction error components}
    $D(\hat{\sigma}, \sigma)$ captures 
    \emph{false positives}: 
    packets the predicted optimal transmits but the true optimal discards, causing the online algorithm to retain suboptimal packets.
    $D(\sigma, \hat{\sigma})$ captures 
    \emph{false negatives}: 
    high-value packets the true optimal transmits but the predicted schedule ignores.
\end{remark}


\section{Algorithms}  
\label{sec3:algorithms}
In this section, we first introduce an intuitive algorithm (Algorithm~\ref{alg:intuitive algorithm}) that achieves optimal consistency.
However, it fails to guarantee a finite competitive ratio bound when predictions are arbitrarily poor. To address this limitation, we refine it to our main 
framework
(Algorithm~\ref{alg:main algorithm}), which integrates a dynamic protective fallback mechanism with a buffer-clearing strategy to guarantee 
strict asymptotic
robustness.

\subsection{1-Consistent, but Non-Robust Algorithm} \label{sec3-1:intutitive algorithm}
Algorithm~\ref{alg:intuitive algorithm} strictly adheres to the optimal schedule of the predicted sequence $\hat{\sigma}$, entirely trusting the prediction without evaluating its accuracy.
\begin{algorithm}[ht]
    \caption{1-Consistent, but Non-Robust Algorithm}
    \label{alg:intuitive algorithm}
    \KwIn{Actual input $\sigma$, predicted input $\hat{\sigma}$}
    Compute the offline optimal schedule $\OPT(\hat{\sigma})$ over the prediction $\hat{\sigma}$\;
    \For{each time step $i \in [T]$}{
        Process arrivals $\sigma_i$ identically to $\OPT(\hat{\sigma})$\;
        Transmit the head packet of 
        the buffer 
        $B_{i}^{\ALG}$\;
    }
\end{algorithm}
\begin{theorem} \label{thm:intuitive consistency}
    Algorithm~\ref{alg:intuitive algorithm} is 1-consistent.
\end{theorem}
\begin{proof} 
    Since Algorithm~\ref{alg:intuitive algorithm} 
    perfectly replicates
    $\OPT(\hat{\sigma})$ throughout 
    the
    execution, we have 
    \begin{equation*} \label{eq:perfect prediction}
        \ALG(\sigma,\hat\sigma) = \OPT(\hat{\sigma}).
    \end{equation*}
    Under perfect predictions ($\hat\sigma = \sigma$), 
    the competitive ratio evaluates to:
	\begin{equation*}
        \sup\limits_{\hat\sigma = \sigma} \frac{v(\OPT(\sigma))}{v(\ALG(\sigma, \hat\sigma))} 
            = \sup\limits_{\hat\sigma 
            = \sigma}\frac{v(\OPT(\sigma))}{v(\OPT(\hat\sigma))} 
            = 1.    \notag
	\end{equation*}
    Hence, Algorithm~\ref{alg:intuitive algorithm} achieves 1-consistency.
\end{proof}

\begin{theorem} 
\label{thm:intuitive robustness}
    Algorithm~\ref{alg:intuitive algorithm} exhibits unbounded robustness.
\end{theorem}
\begin{proof}
    To evaluate the robustness, we first establish the competitive ratio as a function of the prediction error.
    We compute the difference in the total transmitted value by explicitly decomposing $\OPT(\hat{\sigma})$ and $\OPT(\sigma)$ into their intersection and respective set differences. Letting $S^{\mathrm{COM}} = \OPT(\hat\sigma) \cap \OPT(\sigma)$ denote the set of commonly transmitted packets, we expand the schedules as follows:
    \begin{equation*} 
    \label{eq:decompostion}
        \begin{cases}
            \OPT(\hat\sigma) = S^{\mathrm{COM}} \cup D(\hat{\sigma},\sigma),  \\
            \OPT(\sigma) = S^{\mathrm{COM}} \cup D(\sigma,\hat{\sigma}).
        \end{cases}
    \end{equation*}
    Here, as established in Remark~\ref{remark:prediction error components}, $D(\hat\sigma, \sigma)$ and $D(\sigma, \hat\sigma)$ represent the packets 
    transmitted exclusively by $\OPT(\hat\sigma)$ and $\OPT(\sigma)$, respectively. The absolute difference evaluates to:
    \begin{align*}
        \left|v(\OPT(\sigma)) - v(\OPT(\hat{\sigma}))\right| &= \left|v(S^{\mathrm{COM}}\cup D(\sigma, \hat{\sigma})) 
        - v(S^{\mathrm{COM}}\cup D(\hat{\sigma}, \sigma))\right| \\
        &= \left|v(D(\sigma, \hat{\sigma})) - v(D(\hat{\sigma}, \sigma))\right| \\
        &\leq \left|v(D(\sigma, \hat{\sigma}))\right| + \left|v(D(\hat{\sigma}, \sigma))\right| \\
        &= \eta(\sigma, \hat{\sigma}). 
    \end{align*}
    Rearranging the absolute value inequality and
    dividing by $v(\ALG(\sigma, \hat{\sigma}))$ yields:
    \begin{equation} 
    \label{eq:intuitive smoothness}
        \frac{v(\OPT(\sigma))}{v(\ALG(\sigma,\hat{\sigma}))}\le \frac{v(\OPT(\hat\sigma))}{v(\ALG(\sigma,\hat{\sigma}))} + \frac{\eta(\sigma,\hat\sigma)}{v(\ALG(\sigma,\hat{\sigma}))}\triangleq f(\eta(\sigma,\hat{\sigma})).
    \end{equation}
    Equation~\eqref{eq:intuitive smoothness} demonstrates that the competitive ratio is strictly bounded by a non-decreasing function $f(\eta(\sigma,\hat{\sigma}))$. For worst-case predictions where the error approaches infinity, this upper bound evaluates to:
    \begin{align*}
        \lim\limits_{\eta(\sigma,\hat{\sigma})\to\infty} f(\eta(\sigma,\hat{\sigma})) = \frac{v(\OPT(\hat\sigma))}{v(\ALG(\sigma,\hat{\sigma}))}+\lim\limits_{\eta(\sigma,\hat{\sigma})\to\infty}\frac{\eta(\sigma,\hat{\sigma})}{v(\ALG(\sigma,\hat{\sigma}))} = \infty.
    \end{align*}
    Hence, Algorithm~\ref{alg:intuitive algorithm} 
    lacks a finite robustness guarantee under arbitrary predictions.
\end{proof}

\subsection{1-Consistent, Smooth, and $\sqrt{3}$-Robust Algorithm} \label{sec3-2:main algorithm}
\begin{algorithm}[ht]
    \caption{1-Consistent, Smooth, and $\sqrt{3}$-Robust Algorithm}
    \label{alg:main algorithm}
    \KwIn{Actual input $\sigma$, predicted input $\hat{\sigma}$, 
    threshold parameter $\rho \in [1, \sqrt{3}]$}
    $\texttt{SWITCHED} \gets \text{FALSE}$\tcp*{State Flag}
    Compute the offline optimal schedule $\OPT(\hat{\sigma})$ over $\hat{\sigma}$\;
    \For{each time step $i \in [T]$}{
        Simulate $\PG(2+\sqrt{3})$ processing arrivals on virtual buffer $B_{i-1}^{\PG}$\; 
        Transmit the head packet of $B_{i}^{\PG}$ to 
        continuously track
        $v({\PG_{\le i}(\sigma))}$\; \tcp*[f]{Maintain virtual baseline state}
        }
        \uIf(\tcp*[f]{Protective Fallback Evaluation}){\texttt{SWITCHED} is TRUE}{
            Follow $\PG(2+\sqrt{3})$ to process arrivals $\sigma_i$ on 
            physical buffer
            $B_{i-1}^{\ALG}$\; \tcp*[f]{Fallback State}
        }{
            \If{$v(\OPT_{\le i}(\hat{\sigma})) / v(\ALG_{\le i}(\sigma)) > \rho$ or $v(\PG_{\le i}(\sigma)) / v(\ALG_{\le i}(\sigma)) > 1$}{
                $\texttt{SWITCHED} \gets \text{TRUE}$\;
                $B_{i-1}^{\ALG} \gets \emptyset$ \tcp*{Execute buffer-clearing}
                Follow $\PG(2+\sqrt{3})$ to process arrivals $\sigma_i$ on $B_{i-1}^{\ALG}$\;
            } \Else {
                Follow $\OPT(\hat{\sigma})$ to process arrivals $\sigma_i$ on $B_{i-1}^{\ALG}$\;
            }
        }
        Transmit the head packet of $B_{i}^{\ALG}$\;
\end{algorithm}

Algorithm~\ref{alg:main algorithm} extends 
the intuitive algorithm by embedding a real-time protective mechanism.
Upon detecting an untrustworthy prediction, it permanently switches to $\PG(2+\sqrt{3})$, the best-known $\sqrt{3}$-competitive 
baseline 
algorithm. The switch is 
governed by two simultaneous thresholds:
\begin{itemize}
    \item {Prediction-Consistency Check ($v(\OPT_{\le i}(\hat{\sigma})) / v(\ALG_{\le i}(\sigma)) > \rho$):} 
    Detects when the algorithm significantly lags behind the predicted optimal schedule. The parameter $\rho$ regulates strictness. The lower bound $\rho=1$ enforces absolute adherence to yield $1$-consistency, while the upper bound $\rho=\sqrt{3}$ structurally matches the worst-case robustness limit.
    \item {PG-Robustness Check ($v(\PG_{\le i}(\sigma)) / v(\ALG_{\le i}(\sigma)) > 1$):}
    Detects critical false negatives. If the prediction systematically ignores high-value packets, the Prediction-Consistency Check alone is insufficient, as both the predicted optimal and the algorithm would incorrectly discard them, keeping their ratio artificially low. To prevent this, a virtual instance of $\PG(2+\sqrt{3})$ operates continuously on a separate buffer. It serves as a real-time lower bound of baseline performance, triggering an immediate switch if the prediction falls behind what a greedy heuristic would naturally achieve.
\end{itemize}

Upon violating either threshold, Algorithm~\ref{alg:main algorithm} 
executes the \emph{buffer-clearing strategy} and permanently
transitions to the baseline.

\subparagraph{The Buffer-Clearing Strategy.}
\label{sec3-3:buffer-clearing strategy}
Switching to the baseline algorithm during the execution introduces a technical challenge. At the moment of the switch, the current buffer state was constructed based on inaccurate predictions. Continuing to process packets with this suboptimal state could violate the competitive bounds of the $\PG$ algorithm.
Therefore, upon switching, the algorithm executes a buffer-clearing strategy that 
instantaneously discards all resident packets.

While discarding the entire buffer may appear excessively aggressive, the theoretical insight is that the strict capacity constraint naturally bounds the total discarded value. Consequently, the resulting loss is strictly an additive constant that vanishes asymptotically.
Lemma~\ref{lem:OPT comparison} 
formally proves
that the offline optimal value differs by at most 
the initial buffer value
$\varepsilon$ 
when comparing executions starting from an arbitrary buffer versus an empty buffer (the formal proof is deferred to Section~\ref{sec4-2:buffer clearing}).

\begin{lemma}[OPT Comparison Lemma] 
\label{lem:OPT comparison} 
    Consider 
    an arbitrary initial buffer
    $B$ and 
    an empty buffer
    $B' = \emptyset$. 
    Let 
    $v(\OPT(B,\sigma))$ 
    denote 
    the total value transmitted by the optimal schedule
    operating on input $\sigma$ starting from buffer $B$.
    Then:
    \begin{equation} 
        v(\OPT(B', \sigma)) - \varepsilon \le v(\OPT(B, \sigma)) 
        \le v(\OPT(B', \sigma)) + \varepsilon,  \notag
    \end{equation} 
    where $\varepsilon = |v(B)-v(B')| = v(B)$ 
    represents the total resident value of the initial buffer.
\end{lemma} 


\section{Analysis of Algorithm~\ref{alg:main algorithm}}
\label{sec4:analysis}
We now prove that Algorithm~\ref{alg:main algorithm} simultaneously achieves $1$-consistency, $\eta$-smoothness, and \emph{asymptotic $\sqrt{3}$-robustness}. The consistency 
is largely inherited from Algorithm~\ref{alg:intuitive algorithm}, while the smoothness and robustness is guaranteed by the protective mechanism and Lemma~\ref{lem:OPT comparison}.

\subsection{Analysis of Consistency and Smoothness}
\label{sec4-1:smoothness}
\begin{theorem} 
\label{thm:main consistency}
    Algorithm~\ref{alg:main algorithm} is 1-consistent.
\end{theorem}
\begin{proof}
    Under perfect predictions, 
    $v(\OPT_{\le i}(\hat\sigma)) = v(\OPT_{\le i}(\sigma))$
    for all $i \in [T]$.
    Furthermore, the optimal sequence trivially upper-bounds the greedy baseline, ensuring $v(\PG_{\le i}(\sigma)) \le v(\ALG_{\le i}(\sigma))$ holds at every step $i$.
    Therefore, 
    neither the Prediction-Consistency Check nor the PG-Robustness Check is ever triggered.
    Algorithm~\ref{alg:main algorithm} 
    mirrors $\OPT(\sigma)$ precisely, and $1$-consistency follows directly 
    from Theorem~\ref{thm:intuitive consistency}.
\end{proof}

\begin{theorem}
\label{thm:main smoothness}
    Algorithm~\ref{alg:main algorithm} is $\eta$-smooth.
\end{theorem}
\begin{proof}
    Following the analytical steps established
    in the proof of Theorem~\ref{thm:intuitive robustness}, 
    we have the global upper bound:
    \begin{equation}
        v(\OPT(\sigma)) \le v(\OPT(\hat{\sigma})) + \eta(\sigma,\hat\sigma).\notag
    \end{equation}
    We focus on the execution where the algorithm relies on the prediction (i.e., the protective fallback is not triggered).
    Under this condition, 
    the Prediction-Consistency Check
    ensures that the performance strictly satisfies
    \begin{equation}
        v(\OPT(\hat{\sigma})) \le \rho\cdot v(\ALG(\sigma,\hat\sigma)). \notag
    \end{equation}
    Substituting this into the upper bound yields:
    \begin{equation}
        v(\OPT(\sigma)) \le \rho\cdot v(\ALG(\sigma,\hat\sigma)) + \eta(\sigma,\hat\sigma). \notag
    \end{equation}
    Dividing both sides by $v(\ALG(\sigma,\hat\sigma))$
    yields our first bound on the competitive ratio:
    \begin{equation}
        \frac{v(\OPT(\sigma))}{v(\ALG(\sigma,\hat\sigma))} 
        \le \rho + \frac{\eta(\sigma,\hat\sigma)}{v(\ALG(\sigma,\hat\sigma))}. \notag
    \end{equation}
    Simultaneously, because the PG-Robustness Check is never triggered, we guarantee that
    \begin{equation}
        v(\PG(\sigma)) \le v(\ALG(\sigma,\hat\sigma)). \notag
    \end{equation}
    Combining this with the $\sqrt{3}$-competitiveness of $\PG$, we obtain a second bound:
    \begin{equation}
        \frac{v(\OPT(\sigma))}{v(\ALG(\sigma,\hat\sigma))} \le \frac{\sqrt{3}\cdot v(\PG(\sigma))}{v(\ALG(\sigma,\hat\sigma))} \le \sqrt{3}. \notag
    \end{equation}
    Taking the minimum of these two bounds, we conclude that the competitive ratio is globally bounded by the function, which we define as:
    \begin{equation}
        f(\eta(\sigma,\hat\sigma)) \triangleq \min\left( \rho + \frac{\eta(\sigma,\hat\sigma)}{v(\ALG(\sigma,\hat\sigma))}, \, \sqrt{3} \right). \notag
    \end{equation}
    Since $\rho \ge 1$ is a constant, $f(\eta)$ is a non-decreasing function that degrades smoothly with the prediction error $\eta$, and is strictly capped at $\sqrt{3}$. Note that if an excessively large error forces a switch to the baseline, the performance is independently capped by the asymptotic $\sqrt{3}$-robustness (proved subsequently in Theorem~\ref{thm:main robustness}). Hence, Algorithm~\ref{alg:main algorithm} is $\eta$-smooth.
\end{proof}

\subsection{Analysis of the Buffer-Clearing Strategy}
\label{sec4-2:buffer clearing}
To prove Lemma~\ref{lem:OPT comparison}, we 
partition
the transmitted packets into those 
arriving in
the future input $\sigma$ and those 
originally residing in buffer $B$. 
\begin{proposition} 
\label{prop:empty > buffer}
    Let $\OPT_{\emptyset}(\sigma)$ and $\OPT_B(\sigma)$ 
    denote the optimal schedules constructed strictly from packets arriving in the future input sequence $\sigma$, starting from buffer states $\emptyset$ and $B$, respectively.
    Then,
    $v(\OPT_{\emptyset}(\sigma)) \ge v(\OPT_B(\sigma)).\notag$
\end{proposition}
\begin{proof}
    We observe that any sequence of newly arriving packets that can be feasibly scheduled starting from a partially full buffer $B$ also remains feasible when starting from an empty buffer $\emptyset$.
    Let $\OPT_B(\sigma) = (b_1, b_2, \ldots, b_n)$.
    Since $\OPT_\emptyset(\sigma)$ is the optimal schedule over 
    the superset of all such feasible schedules,
    we have $v(\OPT_\emptyset(\sigma)) \ge \sum_{i=1}^n v(b_i) = v(\OPT_B(\sigma))$.
\end{proof}

\begin{proof}[Proof of Lemma~\ref{lem:OPT comparison}]
    We first prove the right-hand inequality. 
    Since $B' = \emptyset$, we have $\varepsilon =|v(B) - v(B')| = v(B)$. We decompose $\OPT(B, \sigma)$ into two disjoint: $\OPT_B(\sigma)$ 
    (transmissions from future input)
    and $\OPT(B, \sigma)\setminus \OPT_B(\sigma)$ 
    (transmissions from the initial buffer).
    Because the latter is drawn entirely from $B$, its value cannot exceed $v(B) = \varepsilon$. Applying Proposition~\ref{prop:empty > buffer}:
    \begin{align}
        v(\OPT(B, \sigma)) &= v(\OPT_B(\sigma)) + v(\OPT(B, \sigma)\setminus \OPT_B(\sigma)) \notag \\
        &\le v(\OPT_{B'}(\sigma)) + v(\OPT(B, \sigma)\setminus \OPT_B(\sigma)) \notag \\ 
        &\le v(\OPT_{B'}(\sigma)) + v(B) \notag \\ 
        &= v(\OPT(B',\sigma)) + \varepsilon. \notag 
    \end{align}
    Next, we prove the left-hand inequality. 
    Consider 
    a suboptimal schedule that immediately preempts every packet in $B$ at time step $1$, subsequently executing $\OPT(\emptyset, \sigma)$. 
    This schedule achieves 
    a value of
    $v(\OPT(\emptyset, \sigma)) = v(\OPT(B', \sigma))$. Since $\OPT(B, \sigma)$ 
    dominates all feasible schedules originating from $B$:
    \begin{equation}
        \OPT(B, \sigma) \ge \OPT(\emptyset,\sigma) = \OPT(B', \sigma) \ge \OPT(B', \sigma)-\varepsilon.\notag
    \end{equation}
\end{proof}


\subsection{Analysis of Robustness}
\label{sec4-3:robustness}
\begin{theorem}
\label{thm:main robustness}
    Algorithm~\ref{alg:main algorithm} is asymptotically $\sqrt{3}$-robust.
\end{theorem}
\begin{proof}
    Assume the algorithm triggers a switch at time step $t^*$. We analyze the execution by partitioning the timeline into $[1, t^*-1]$ and $[t^*, T]$.
    Prior to $t^*$, the PG-Robustness Check remained unviolated, guaranteeing $v(\PG_{\le i}(\sigma)) \leq v(\ALG_{\le i}(\sigma, \hat\sigma))$ for all $i < t^*$. Combining this with the established $\sqrt{3}$-competitiveness of the $\PG$ algorithm:
    \begin{equation}
        v(\OPT_{<t^*}(\sigma)) \le \sqrt{3} \cdot v(\PG_{<t^*}(\sigma)) \le \sqrt{3} \cdot v(\ALG_{<t^*}(\sigma, \hat\sigma)). \notag
    \end{equation}

    
    From step $t^*$ onward, the physical buffer is cleared. Applying Lemma~\ref{lem:OPT comparison} to relate the optimal schedule starting from $B_{t^*-1}^{\OPT}$ to an empty buffer, and invoking the $\sqrt{3}$-competitiveness of $\PG$:
    \begin{align}
        v(\OPT_{\ge t^*}(\sigma)) &= v(\OPT(B_{t^*-1}^{\OPT}, \sigma_{\ge t^*})) \notag\\
        &\le v(\OPT(\emptyset,\sigma_{\ge t^*})) + v(B_{t^*-1}^{\OPT}) \notag \\
        &\le \sqrt{3} \cdot v(\PG(\emptyset,\sigma_{\ge t^*})) + v(B_{t^*-1}^{\OPT}) \notag \\
        &= \sqrt{3} \cdot v(\ALG_{\ge t^*}(\sigma, \hat\sigma)) + v(B_{t^*-1}^{\OPT}). \notag
    \end{align}

    Summing the partitions to compute the global competitive ratio:
    \begin{align}
        \frac{v(\OPT(\sigma))}{v(\ALG(\sigma, \hat\sigma))} &= \frac{v(\OPT_{< t^*}(\sigma)) + v(\OPT_{\ge t^*}(\sigma))}{v(\ALG_{< t^*}(\sigma, \hat\sigma)) + v(\ALG_{\ge t^*}(\sigma, \hat\sigma))} \notag\\
        &\le \frac{\sqrt{3} \cdot \left[ v(\ALG_{< t^*}(\sigma, \hat\sigma)) + v(\ALG_{\ge t^*}(\sigma, \hat\sigma)) \right] + v(B_{t^*-1}^{\OPT})}{v(\ALG(\sigma, \hat\sigma))} \notag\\
        &= \sqrt{3} + \frac{v(B_{t^*-1}^{\OPT})}{v(\ALG(\sigma, \hat\sigma))}. \notag
    \end{align}
    
    Evaluating the limit as $v(\ALG(\sigma, \hat\sigma))$ approaches infinity:
    \begin{equation}
        \lim\limits_{v(\ALG(\sigma, \hat\sigma)) \to \infty} \frac{v(\OPT(\sigma))}{v(\ALG(\sigma, \hat\sigma))} \le \sqrt{3} + \lim\limits_{v(\ALG(\sigma, \hat\sigma)) \to \infty} \frac{v(B_{t^*-1}^{\OPT})}{v(\ALG(\sigma, \hat\sigma))} = \sqrt{3}. \notag
    \end{equation}
    The additive capacity constant vanishes, establishing that Algorithm~\ref{alg:main algorithm} is asymptotically $\sqrt{3}$-robust.
\end{proof}

The properties of Algorithm~\ref{alg:main algorithm} operate entirely independently of the specific baseline chosen. 
The $1$-consistency holds 
natively without a switch,
$\eta$-smoothness 
functions globally via the error metric, and the buffer-clearing loss is uniformly bounded across any algorithm. Consequently, substituting
$\PG$ with any $\beta$-competitive algorithm 
immediately yields a generalized framework.

\begin{corollary}
\label{cor:generalized framework}
    For any $\beta$-competitive algorithm 
    $\mathcal{A}$, let $\mathcal{R}(\mathcal{A})$ be the 
    generalized algorithm derived from 
    Algorithm~\ref{alg:main algorithm} by 
    applying the following modifications:
    \begin{itemize}
        \item Replace $\PG$ with $\mathcal{A}$ as the fallback 
        algorithm.
        \item Adapt the PG-Robustness Check to 
        continuously track $v(\mathcal{A}_{\le i}(\sigma))$ in place of $v(\PG_{\le i}(\sigma))$.
        \item 
        Expand the admissible range of the threshold parameter
        $\rho$ to $[1, \beta]$.
    \end{itemize}
    Then, $\mathcal{R}(\mathcal{A})$ achieves $1$-consistency, $\eta$-smoothness, and asymptotic $\beta$-robustness.
\end{corollary}
\begin{proof}
    The $1$-consistency and $\eta$-smoothness of $\mathcal{R}(\mathcal{A})$ follow directly from Theorems~\ref{thm:main consistency} 
    and~\ref{thm:main smoothness}, 
    as these properties hold independently of the underlying fallback mechanism.
    The proof of asymptotic $\beta$-robustness proceeds analogously
    to that of Theorem~\ref{thm:main robustness}, 
    substituting
    the $\sqrt{3}$-competitiveness of $\PG$ with the $\beta$-competitiveness of $\mathcal{A}$.
\end{proof}


\section{Conclusion and Open Problems}
\label{sec5:conclusion and open problems}
We bridge the gap between classical worst-case limits and modern machine-learned predictions for
the preemptive FIFO buffer management problem.
Specifically, we present a novel learning-augmented framework that simultaneously achieves 
$1$-consistency, $\eta$-smoothness, and asymptotic $\sqrt{3}$-robustness. 
By continuously monitoring prediction quality against a virtual baseline, the algorithm dynamically executes a buffer-clearing fallback strategy when predictions prove unreliable. Crucially, we prove that the structural loss incurred by this clearing operation is strictly bounded by the buffer capacity, an additive constant that vanishes asymptotically, thereby preserving the worst-case $\sqrt{3}$-robustness.

A central technical innovation of this research is the formulation of an output-based prediction error metric. Recognizing that traditional input-based distances artificially penalize mispredictions on untransmitted packets in capacity-constrained environments, the proposed metric assesses prediction quality strictly over the resulting optimal schedules. By grounding the error in the symmetric difference of the outputs, this approach precisely captures the true cost of prediction failures, which directly enables the graceful $\eta$-smoothness guarantee. Furthermore, this mechanism naturally extends into a generalized framework: substituting the fallback module with any $\beta$-competitive online algorithm seamlessly yields an asymptotic $\beta$-robustness guarantee.

These results open several promising directions for future research. A fundamental open question is to establish the theoretical lower bound on the robustness of strictly $1$-consistent algorithms for this problem, characterizing the inherent tradeoff between perfect consistency and worst-case robustness. Additionally, designing practical machine learning models explicitly tailored to minimize the proposed output-based error metric, or exploring the theoretical limits of alternative prediction models (e.g., predicting packet values or arrival distributions), remains an exciting avenue for bridging theoretical guarantees with empirical network performance.




\bibliography{ESA_arxiv}

\end{document}